\documentclass[11pt,a4paper]{article}

\pdfoutput=1

\newcommand{\mytitle}{Pointer Quantum PCPs and Multi-Prover Games}
\newcommand{\Alex}{Alex B. Grilo}
\newcommand{\Iordanis}{Iordanis Kerenidis}
\newcommand{\Attila}{Attila Pereszl\'{e}nyi}

\usepackage{hyperref}
\usepackage{fullpage}
\usepackage{amssymb}
\usepackage{mathtools}
\usepackage{amsthm}
\usepackage{engord}
\usepackage{bookmark}
\usepackage{flafter}
\usepackage[nodayofweek]{datetime}
\usepackage{algorithm}
\usepackage{eucal}
\usepackage{dsfont}
\usepackage[T1]{fontenc}
\usepackage{lmodern}
\usepackage{authblk}
\usepackage{color,graphicx}
\usepackage{tikz}
\usepackage{enumitem}
\usepackage{subcaption}
\usetikzlibrary{positioning, arrows,calc, fit, shapes, chains, matrix}
\usepackage[stretch=10,shrink=10]{microtype}
\usepackage[capitalise]{cleveref}

\hypersetup{
	pdfstartview={FitH},
	pdfdisplaydoctitle={true},
	breaklinks={true},
	colorlinks={true},
	linkcolor={black},
	citecolor={black},
	urlcolor={black},
	bookmarksopen={true},
	bookmarksnumbered={false},
	pdftitle={\mytitle},
	pdfauthor={\Alex, \Iordanis, \Attila}
}

\allowdisplaybreaks[1]

\theoremstyle{plain}
\newtheorem{theorem}{Theorem}[section]
\newtheorem{lemma}[theorem]{Lemma}

\newtheorem{conjecture}[theorem]{Conjecture}
\newtheorem{definition}[theorem]{Definition}

\theoremstyle{definition}

\floatname{algorithm}{Protocol}
\crefname{algorithm}{Protocol}{Protocols}
\crefname{conjecture}{Conjecture}{Conjectures}

  \newcommand{\comment}[1]{}

  \newcommand{\class}[1]{\ensuremath{\mathsf{#1}}}

  \newcommand{\half}{\ensuremath{\frac{1}{2}}}

  \newcommand{\calS}{\mathcal{S}}
  
  \newcommand{\calH}{\mathcal{H}}
  
  \newcommand{\Zplus}{\ensuremath{\mathbb{Z}^+}}

  \newcommand{\kb}[1]{\ketbra{#1}}

  \newcommand{\cprob}[2]{\prob{#1 \middle\vert #2}}

  \newcommand{\ayes}{A_{\rm yes}}
  \newcommand{\ano}{A_{\rm no}}

  \newcommand{\complex}{\mathbb{C}}
  \newcommand{\real}{\mathbb{R}}

  \newif\ifcomments
  \commentstrue
  \ifcomments
    \newcommand{\knote}[1]{\textcolor{cyan}{ {\textbf{(Iordanis: }#1\textbf{) }}}}
    \newcommand{\gnote}[1]{\textcolor{blue}{ {\textbf{(Alex: }#1\textbf{) }}}}
    \newcommand{\gfootnote}[1]{\textcolor{blue}{\footnote{\textcolor{blue}{\textbf{Alex: }#1}}}}
    \newcommand{\pnote}[1]{\textcolor{red}{ {\textbf{(Attila: }#1\textbf{) }}}}
    \newcommand{\pfootnote}[1]{\textcolor{red}	{\footnote{\textcolor{red}{\textbf{Attila: }#1}}}}
  \else
    \newcommand{\knote}[1]{}
    \newcommand{\gnote}[1]{}
    \newcommand{\gfootnote}[1]{}
    \newcommand{\pnote}[1]{}
    \newcommand{\pfootnote}[1]{}
  \fi

  \newif\ifhighlights
  \highlightstrue
  \ifhighlights

  \else

  \fi

  \newcommand{\NP}{\class{NP}}
  
  \newcommand{\QMA}{\class{QMA}}

  \newcommand {\LH} [3] {\fn{\text{\textnormal{\textsc{LocalHam}}}}{#1, #2, #3}}
  \newcommand {\SLH} [3] {\fn{\text{\textnormal{\textsc{SLH}}}}{#1, #2, #3}}

  \newcommand {\QPCP} [3] {\fn{\class{QPCP}}{#1, #2, #3}}
  \newcommand {\PQPCP} [3] {\fn{\class{PointerQPCP}}{#1, #2, #3}}

  \newcommand {\CRESP} [3] {\fn{\text{\textnormal{CRESP}}}{#1, #2, #3}}

\newcommand {\br} [1] {\ensuremath{ \left( #1 \right) }}
\newcommand {\Br} [1] {\ensuremath{ \left[ #1 \right] }}
\newcommand {\cbr} [1] {\ensuremath{ \left\lbrace #1 \right\rbrace }}
\newcommand {\minusspace} {\: \! \!}

\newcommand {\fn} [2] {\ensuremath{ #1 \minusspace \br{ #2 } }}
\newcommand {\Fn} [2] {\ensuremath{ #1 \minusspace \Br{ #2 } }}
\newcommand {\eqdef} {\ensuremath{ \stackrel{\mathrm{def}}{=} }}
\newcommand {\prob} [1] {\Fn{\Pr}{#1}}
\newcommand {\abs} [1] {\ensuremath{ \left| #1 \right| }}
\newcommand {\norm} [1] {\ensuremath{ \left\| #1 \right\| }}

\newcommand {\footnoteEmail} [1] {\thanks{\href{mailto:#1}{\texttt{#1}}}}
\newcommand {\bigo} [1] {\fn{O}{#1}}
\newcommand {\Epsilon} {\mathcal{E}}
\newcommand {\bigomega} [1] {\fn{\Omega}{#1}}
\newcommand {\bra} [1] {\ensuremath{ \left\langle #1 \right| }}
\newcommand {\ket} [1] {\ensuremath{ \left| #1 \right\rangle }}
\newcommand {\ketbratwo} [2] {\ensuremath{ \left| #1 \middle\rangle \middle\langle #2 \right| }}
\newcommand {\ketbra} [1] {\ketbratwo{#1}{#1}}

\newcommand {\ceil} [1] {\ensuremath{ \left\lceil #1 \right\rceil }}

\newcommand {\Tr} {\ensuremath{ \mathrm{Tr} }}

\newcommand {\id} {\ensuremath{\mathds{1}}}

\newcommand {\myqedhere} {\tag*{\qedhere}}

\newcommand {\trace} [1] {\fn{\Tr}{#1}}

\newcommand {\poly} [1] {\fn{\text{\textnormal{poly}}}{#1}}
\newcommand {\Hset} [1] {\ensuremath{\mathbf{#1}}}
\newenvironment{proofof}[1]
	{\begin{proof}[Proof of \cref{#1}]}
	{\end{proof}}
\newenvironment{proofsketch}
	{\begin{proof}[Proof Sketch]}
	{\end{proof}}

\tikzset{
  table/.style={
  matrix of nodes,
  row sep=-\pgflinewidth,
  column sep=-\pgflinewidth,
  nodes={rectangle,draw=black,text width=1.25ex,align=center},
  text depth=0.25ex,
  text height=1ex,
  nodes in empty cells
  },
  texto/.style={font=\footnotesize\sffamily},
  title/.style={font=\small\sffamily}
}

\newcommand\CellText[2]{%
  \node[texto,left=-0.5ex of mat#1,anchor=east]
  at (mat#1.west)
  {\textnormal{#2}};
}

\newcommand\SlText[2]{%
  \node[texto,above=-0.4ex of mat#1,anchor=west,rotate=90]
  at (mat#1.north)
  {\textnormal{#2}};
}
\newcommand\SlTextNormal[2]{%
  \node[texto,above=2.5ex of mat#1,anchor=north]
  at (mat#1.north)
  {\textnormal{#2}};
}

  \author[1]{\Alex\footnoteEmail{abgrilo@gmail.com}}
  \author[1,2]{\Iordanis\footnoteEmail{jkeren@gmail.com}}
  \author{\Attila\footnoteEmail{attila.pereszlenyi@gmail.com}}
  \affil[1]{IRIF, CNRS, Universit\'{e} Paris Diderot, Paris, France}
  \affil[2]{Centre for Quantum Technologies, National University of Singapore,
  Singapore}

  \date{\formatdate{1}{3}{2016}}
  \title{\textbf{\mytitle}}

\begin{document}
\maketitle

\begin{abstract}
  The quantum PCP (QPCP) conjecture states that all problems in \QMA{}, the quantum analogue of \NP, admit quantum
  verifiers that only act on a constant number of qubits of a polynomial size quantum proof and have a constant gap between completeness and soundness. Despite an impressive body of work trying to prove or disprove the quantum PCP conjecture, it still remains widely open. The above-mentioned  proof verification statement has also been shown equivalent to the \QMA{}-completeness of the Local Hamiltonian problem with constant relative gap. Nevertheless, unlike in the classical case, no equivalent formulation in the language of multi-prover games is known.

  In this work, we propose a new type of quantum proof systems,
  the Pointer QPCP, where a verifier first accesses a classical proof that he can use as a pointer to which qubits from the quantum part of the proof to access.
  We define the Pointer QPCP conjecture, that states that all problems in \QMA{} admit quantum verifiers that first access a logarithmic number of bits from the classical part of a polynomial size proof, then act on a constant number of qubits from the quantum part of the proof, and have a constant gap between completeness and soundness.
  We define a new \QMA{}-complete problem, the Set Local Hamiltonian problem,
  and a new restricted class of quantum multi-prover games, called CRESP games.
  We use them to provide two other equivalent statements to the Pointer QPCP conjecture: the Set Local Hamiltonian problem with constant relative gap is \QMA{}-complete; and the
  approximation of the maximum acceptance probability of CRESP games up
  to a constant additive factor is as hard as
  \QMA{}. This is the first equivalence between a quantum PCP statement and the inapproximability of quantum multi-prover games.
\end{abstract}


\section{Introduction}

The celebrated PCP theorem states that all languages in \NP{} can be verified probabilistically by randomized verifiers
  that only check a constant number of bits of a polynomial size proof \cite{AroraLMSS98,AroraS98,Dinur07}.
  This theorem has far-reaching applications in complexity theory and especially
  in the inapproximability of certain optimization
  problems. This is because the PCP theorem can be recast in the following equivalent way: the
  approximation of MAX-SAT up to some constant additive factor is \NP{}-complete.
  Let us also remark that the classical PCP theorem has a third very interesting
equivalent formulation as approximation of the maximum acceptance probability of
some polynomial size multi-prover interactive games \cite{Raz98}. This game
formulation was fundamental in order to achieve better constants for the
inapproximability results of a number of \NP{}-hard problems.

One of the main questions in quantum complexity theory is whether one can prove an analogous statement for the class \QMA{}, the quantum analogue of \NP{}. The QPCP conjecture \cite{AharonovAV13} has received a lot of attention due to its importance to both physics and theoretical computer science and an impressive body of work has
provided either positive evidence
\cite{AharonovALV09,EldarH15,NatarajanV15} or
negative \cite{Arad11,AharonovE15,BrandaoH13a}. There are many different
ingredients that go into the proof of the classical PCP theorem, especially
since there are two different ways of proving it, one through the proof system
formulation and another more combinatorial way by looking directly at the
inapproximability of constraint satisfaction problems.
In the quantum setting, the positive and negative
evidence has been mostly that certain techniques that had been used in the
classical setting are applicable or not in the quantum setting.
We note that
the lack of a way of seeing the quantum PCP conjecture in a game context
also prevents us from using some
important techniques that are present in the classical case,
such as the parallel repetition theorem. Overall,
proving the quantum PCP theorem remains a daunting task.

The QPCP conjecture can be cast as a type of proof system which we denote by \QPCP{q}{\alpha}{\beta}.
Here a quantum verifier tosses a logarithmic number of classical coins and, based on the coin outcomes, decides on which $q$ qubits from the polynomial-size quantum proof
to perform a measurement. The measurement output decides on acceptance or rejection.
A yes instance is accepted with probability at least $\alpha$ and a no instance
is accepted with probability at most $\beta$,
for some $\alpha > \beta$ \cite{AharonovALV09,AharonovAV13}.
The formal conjecture is stated below.

\begin{conjecture}[QPCP Conjecture - Proof verification version]
	\label{conj:QPCP verification}
	$\QMA = \QPCP{q}{\alpha}{\beta}$ where
	$q = \bigo{1}$ and $\alpha - \beta = \bigomega{1}$.
\end{conjecture}

In quantum mechanics, the evolution of quantum systems are described by Hermitian
operators called Hamiltonians.
In nature, particles that are far apart tend
not to interact so the global Hamiltonian can usually be described as a
sum of local Hamiltonians.
The Local Hamiltonian problem, denoted by \LH{k}{a}{b}, receives as input $m$
Hamiltonians $H_1,\ldots,H_m$ where each one describes the evolution of at
most $k$ qubits and the question is if there is a global state such that its energy is
at most $am$ or all states have energy at least $bm$ for $b>a$.
The area studying the above problem is called quantum Hamiltonian
complexity \cite{Osborne12, GharibianHLS15}, a topic that lies in the
intersection of physics, computer science and mathematics.
It began with Kitaev who showed that for $b - a \geq 1/\poly{n}$,
\LH{5}{a}{b} is complete for the class \QMA{} \cite{AharonovN02, KitaevSV02}.
It has subsequently been
improved, reducing the locality of the Hamiltonians to two \cite{KempeKR06} and restricting their
structure \cite{KR03, OliveiraT10, CubittM14, HallgrenNN13}.
These results imply that estimating the groundstate energy of a system
within an inverse polynomial additive factor is hard.
It is natural to ask if it still remains hard if we require only
constant approximation.
The physical interpretation of this problem is connected to
the stability of entanglement in ``room
temperature''.

The second equivalent statement of the quantum PCP conjecture asks
if \LH{k}{a}{b} remains \QMA{}-complete when $b -a$ is constant.
It is stated formally in the conjecture below.

\begin{conjecture}[QPCP Conjecture - Constraint satisfaction version]
	\label{conj:QPCP constraint}
	The Local Hamiltonian problem \LH{k}{a}{b} is \QMA-complete for $k=O(1)$ and
	$b - a = \bigomega{1} $,
	where the \QMA-hardness is with respect to quantum reductions.
\end{conjecture}

The two versions of the quantum PCP conjecture have been proven equivalent \cite{AharonovAV13}, and since \cref{conj:QPCP constraint} is true for $b - a  \geq 1/\poly{n}$, we can also conclude that $\QMA = \QPCP{q}{\alpha}{\beta}$ with $q = \bigo{1}$ and $\alpha - \beta \geq 1/\poly{n}$.

Let us note that so far there is no equivalent statement of the QPCP
conjecture in the language of multi-prover games, though the approximation of
the maximum acceptance probability of certain multi-prover games up
to an inverse-polynomial additive factor has been proven to be \QMA{}-hard
\cite{FitzsimonsV15}.

\subsection{Our Results}

 In our work, we propose a new type of quantum proof systems, the Pointer QPCP,
and formulate three equivalent versions of the Pointer  QPCP conjecture.
This may help towards proving or disproving the original  QPCP conjecture.
We start by describing a new proof verification system then we provide
a new variant of the Local Hamiltonian problem and last we
describe an equivalent polynomial size multi-prover game.
Up to our knowledge, this is the first time a polynomial size
multi-prover game has been proven equivalent to some  QPCP conjecture.
Our new conjecture is a weaker statement than the original  QPCP conjecture
and hence may be easier to prove.
Moreover, having an equivalent game version of it might also lead to new methods that could potentially be relevant
for attacking the original conjecture as well.

We now give some details of our results. We define a new quantum proof system, where the proof contains two separate parts,
a classical and a quantum proof both of polynomial size.
The verifier can first access a logarithmic number of bits from
the classical proof and, depending on the content, he can then access a constant number
of qubits from the quantum proof.
Since the classical part can be seen as a pointer to the qubits that
will be accessed, we denote this proof system by \PQPCP{q}{\alpha}{\beta}.
To be more specific, the verifier first reads a logarithmic number of bits
from the classical part of the proof and then measures at most $q$ qubits from
the quantum part.
He accepts a yes instance with probability at least $\alpha$ and a no instance
with probability at most $\beta$.
Since a Pointer QPCP is a generalization of QPCP, it follows that all
problems in \QMA{} have a \PQPCP{q}{\alpha}{\beta} proof system with
$\alpha - \beta \geq 1/\poly{n}$.

\begin{conjecture}[Pointer QPCP Conjecture - Proof verification  version]
	\label{conj:weak QPCP verification}
	It holds that $\QMA = \PQPCP{q}{\alpha}{\beta}$ where
	$q = \bigo{1}$ and $\alpha - \beta=\bigomega{1}$.
\end{conjecture}

We note that quantum proof systems with classical and quantum parts
have also appeared in \cite{Raz05}. There, the aim was to reduce the number of blocks being read in classical PCPs and hence, in the proposed model, a
logarithmic size quantum proof is provided to the verifier who measures it
and then reads only a single block from a polynomial size classical proof.

In addition to Pointer QPCPs, we also propose a ``constraint satisfaction''
version of the above conjecture which will turn out to be equivalent.  We do
this by defining a new variant of the Local Hamiltonian problem which we call
the Set Local Hamiltonian problem.  Here the input is $m$ sets of a polynomial
number of $k$-local Hamiltonians each, and we ask if there exists a
  representative Hamiltonian from each set such that the Hamiltonian
  corresponding to their
  sum has groundstate energy at most $am$ or for every possible choice of representative
  Hamiltonians from each set, the Hamiltonian corresponding to their sum has groundstate energy at least
$bm$.  We denote the above problem by \SLH{k}{a}{b}.
Since the Local
Hamiltonian problem is a special case of the Set Local Hamiltonian problem,
where the sets are singletons, $\SLH{k}{a}{b}$ is \QMA{}-complete for $k \geq 2$
and $b - a \geq 1/\poly{n}$.

\begin{conjecture}[Pointer QPCP Conjecture - Constraint satisfaction version]
	\label{conj:weak QPCP constraint}
	The \SLH{k}{a}{b} problem is \QMA{}-complete for $k = \bigo{1}$
	and $b - a = \bigomega{1}$.
\end{conjecture}

\begin{figure}[t b]
  \begin{center}
\begin{tikzpicture}
  \coordinate (P) at (-1,0);
  \coordinate (S) at (1,0);
  \coordinate (G) at (0,1.65);

  \coordinate (LG) at (-1.5,0.2);
  \coordinate (RS) at (1.6,0.3);
  \coordinate (LP) at (-0.1,1.73);
  \coordinate (RP) at (0.1,1.72);

  \draw [->, thick] (RS) -- (RP);
  \draw [->, thick] (LP) -- (LG);
  \draw [<-, thick] (S) -- (P);

  \node[label={right: \Cref{conj:weak QPCP constraint} (Set Local Hamiltonians)}] at (S) {};
  \node[label={above: \Cref{conj:weak QPCP game} (CRESP games)}] at (G) {};
  \node[label={left: \Cref{conj:weak QPCP verification}
	(Pointer QPCPs)}] at (P) {};
\end{tikzpicture}
\end{center}
\caption{Arrows of implications in our proof of the equivalence
	of the three conjectures.}
\label{fig:arrows of implications}
\end{figure}
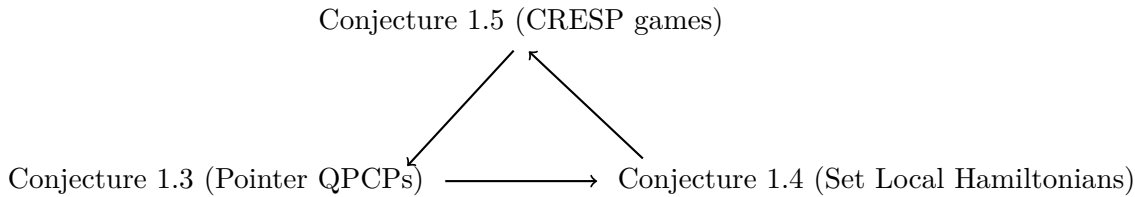

As mentioned earlier, the classical PCP theorem has another interesting
equivalent formulation regarding the approximation of the
maximum acceptance probability of multi-prover games \cite{Raz98}, while the same is not known for the quantum case.
We propose an equivalent multi-prover game formulation of the Pointer
QPCP conjecture.
Our game, which we call CRESP (Classical and Restricted-Entanglement Swapping-Provers) game,
was inspired by the work of Fitzsimons and Vidick \cite{FitzsimonsV15}.
However, in order to prove an equivalence, we had to drastically change the game. In their work, a multi-prover game
is proposed for the Local Hamiltonian problem in which the completeness-soundness
gap is inverse polynomial. If we try to follow the same proof but with an instance of the Local Hamiltonian
with constant gap, the gap does not survive and at the end there will be an
inverse-polynomial gap in the game. Hence we are not able to prove the equivalence with the standard  QPCP conjecture.

We define our CRESP game to have one classical prover and logarithmically
many quantum provers who are restricted both in the strategies they can
perform and also in the initial quantum state they share.
The verifier asks a single question of logarithmic length to all of them, the
classical prover replies with logarithmically many bits, while the quantum
provers reply with $k$ $4$-dimensional qudits. (For simplicity,
we will omit the dimension of the qudit system in the rest of the paper.)
The promise problem \CRESP{k}{\alpha}{\beta} informally asks if we can
distinguish between the cases when the provers win the game with
probability at least $\alpha$ or at most $\beta$.
Similarly to the previous problems, we will see that \CRESP{k}{\alpha}{\beta} is
\QMA{}-complete for $\alpha - \beta \geq 1/\poly{n}$.
See \cref{lemma:CRESPcomplete} for the precise statement.

\begin{conjecture}[Pointer QPCP Conjecture - Game version]
	\label{conj:weak QPCP game}
	The \CRESP{k}{\alpha}{\beta} problem is \QMA{}-complete for $k = \bigo{1}$
	and $\alpha - \beta = \bigomega{1}$.
\end{conjecture}

Our main result is the equivalence of the above three
formulations of the Pointer QPCP conjecture.
It is stated formally in the following theorem.

\begin{theorem}[Main theorem]
	\label{thm:main}
    The three versions of the Pointer QPCP conjecture
    (\cref{conj:weak QPCP verification,conj:weak QPCP constraint,conj:weak QPCP game}) are either all true or all false.
\end{theorem}

The proof is divided into three steps:
first, we show that \cref{conj:weak QPCP verification}
implies \cref{conj:weak QPCP constraint}; second, we show that \cref{conj:weak QPCP constraint}
implies \cref{conj:weak QPCP game}; and finally, we prove that \cref{conj:weak QPCP game}
implies \cref{conj:weak QPCP verification}.
The arrows of implications are depicted in \cref{fig:arrows of implications}.

The paper is organized as follows:
In \cref{sec:Preliminaries}, we describe some standard definitions required for the rest of the paper.
In \cref{sec:Definitions}, we present the definitions of our new notions, the Pointer  QPCPs, the Set Local Hamiltonian problem, and the CRESP games.
The proof of equivalence is presented in \cref{sec:equivalence}. We conclude the paper with some discussion and open problems
in \cref{sec:open problems}.


\section{Preliminaries}
\label{sec:Preliminaries}
In this section we provide some definitions that we use in the paper.
We start by defining \QMA{}, the quantum analogue of \NP{}.

  \begin{definition}[Quantum Merlin-Arthur proof systems]
	\label{def:QMA}
  Let $n \in \Zplus$ be the input size and $p$ be a polynomial. A QMA protocol proceeds in the following steps.
  \begin{enumerate}
    \item The verifier receives an input $x$ and a quantum proof $\ket{\psi}$
		of size $p(n)$.
  \item The verifier runs in polynomial time in $n$. He performs a general POVM
    measurement on $\ket{\psi}$ and
    decides on the acceptance or rejection of the input.
  \end{enumerate}
  A promise problem $A=(\ayes,\ano)$  belongs to \QMA{} if it has a
  QMA proof system with the following properties.
	\begin{description}
    \item[Completeness.] If $x \in \ayes$ then there is a $\ket{\psi}$ such that
      the verifier accepts with
      probability at least $\frac{2}{3}$.
    \item[Soundness.] If $x\in\ano$ then for all $\ket{\psi}$ the
      verifier accepts with probability at most $\frac{1}{3}$.
	\end{description}
  \end{definition}

Now we present the Local Hamiltonian problem, the quantum analogue of MAX-SAT.

\begin{definition}
  \label{def:local Hamiltonian}
  The \emph{Local Hamiltonian} problem is denoted by \LH{k}{a}{b}
  where $k \in \Zplus$ is called the locality and for $a,b \in \real$
  it holds that $a<b$.
  It is the following promise problem. Let $n$ be the number of qubits of a quantum system.
  The input is a set of $m(n)$ Hamiltonians $H_1, \ldots, H_{m(n)}$
  where $m$ is a polynomial in $n$, $\forall i \in \Br{m(n)} : 0 \leq H_i \leq \id $
  and each $H_i$ acts on $k$ qubits out of the $n$ qubit system.
  For $H \eqdef \sum_{j = 1}^{m(n)} H_j$ the following two conditions hold.
    \begin{itemize}
      \item In a YES instance there exists a
      state $\ket{\varphi} \in \complex^{2^{n}}$ such that
      $
        \bra{\varphi} H \ket{\varphi}
        \leq a \cdot m(n) .
      $
      \item In a NO instance for all states $\ket{\varphi} \in \complex^{2^{n}}$
      it holds that
      $
        \bra{\varphi} H \ket{\varphi}
        \geq b \cdot m(n) .
      $
    \end{itemize}
\end{definition}

Kitaev proved that for $k \geq 5$ and $b - a \geq 1/\poly{n}$
the \LH{k}{a}{b} problem is \QMA{}-complete \cite{KitaevSV02}.
This completeness result was later improved for $k \geq 2$ \cite{KempeKR06}.

We now define the quantum analogue of PCPs, a quantum proof system
where the verifier only checks a few
qubits from the quantum proof.

\begin{definition}[Quantum Probabilistically Checkable Proofs]
  \label{def:qpcp}
  Let $n \in \Zplus$ be the input size and $p$ be a polynomial. A QPCP protocol proceeds in the following steps.
  \begin{enumerate}
    \item The verifier receives an input $x$ and a quantum proof $\ket{\psi}$
		of size $p(n)$.
  \item The verifier runs in time polynomial in $n$. He picks $O(\log{n})$ bits
    uniformly at random, and based on the input and on the random bits, he
    performs a general POVM measurement on  $q$ qubits, and decides on acceptance or rejection of
    the input.
  \end{enumerate}
  A promise problem $A=(\ayes,\ano)$  belongs to \QPCP{q}{\alpha}{\beta} if it has a
  QPCP proof system with the following properties.
	\begin{description}
    \item[Completeness.] If $x \in \ayes$ then there is a $\ket{\psi}$ such
      that the verifier accepts with probability at least $\alpha$.
    \item[Soundness.] If $x\in\ano$ then for all $\ket{\psi}$ the
      verifier accepts with probability at most $\beta$.
	\end{description}
\end{definition}

We can easily prove the following statement.
\begin{lemma}
	\label{rem:QMA and QPCP}
	It holds that $\QMA = \QPCP{q}{\alpha}{\beta}$ where
	$q = \bigo{1}$ and $\alpha - \beta \geq 1/\poly{n}$.
\end{lemma}
\begin{proof}
	The containment $\QPCP{q}{\alpha}{\beta} \subseteq \QMA$ is trivial
	since the \QMA{} verifier can read the whole proof and the power
	of \QMA{} doesn't change if the gap is inverse-polynomial.
	The other direction of the containment follows from Kitaev's proof
	that the 5-Local Hamiltonian is \QMA-complete
	\cite{KitaevSV02,AharonovN02}.
	The \QMA{} verifier in the proof is also a QPCP verifier.
\end{proof}

The quantum PCP conjecture has two equivalent versions:

\begin{theorem}[\cite{AharonovAV13}]
	The class $\QMA$ is equal to the class $\QPCP{q}{\alpha}{\beta}$ with
	$q = \bigo{1}$ and $\alpha - \beta = \bigomega{1}$ (\cref{conj:QPCP verification}) if and only if
	the Local Hamiltonian problem \LH{k}{a}{b} is \QMA-complete for $k=O(1)$ and
	$b - a = \bigomega{1} $,
	where the \QMA-hardness is with respect to quantum reductions (\cref{conj:QPCP constraint}).
\end{theorem}


\section{Pointer QPCPs, Set Local Hamiltonians, and CRESP Games}
\label{sec:Definitions}

In this section, we present the definitions required for our conjectures.
We start by defining Pointer QPCPs, a generalized version of
QPCPs, in which the verifier can read a small number of bits from the
classical part of the proof and then, based on that, read a constant number
of qubits from the quantum part of the proof.
Then we propose the Set Local Hamiltonian problem
that can be thought of as a ``constraint satisfaction'' version of the conjecture.
Finally, we define CRESP games which are restricted multi-prover games
for which approximation of their value will turn out to be
equivalent to the other two formulations.


\subsection{Pointer QPCPs}

\begin{definition}
    \label{def:PQPCP}
    Let $n \in \Zplus$ be the input size and let $m,l,p$ be polynomials.
    A Pointer QPCP protocol proceeds in the following steps.
  \begin{enumerate}
	\item The verifier receives an input $x$ and a two-part proof
		of size $m(n)+p(n)$ in the form
        \[y_1...y_{m(n)} \otimes \ket{\psi}\]
        where $y_i \in [l(n)]$ (i.e., each $y_i$ can be  written with \bigo{\log n} bits)
        and \ket{\psi} is a state of $p(n)$ qubits.
        We refer
        to $y_1...y_{m(n)}$ as the classical part of the proof and $\ket{\psi}$ as
        the quantum part of the proof.
  \item The verifier runs in time polynomial in $n$.
    He chooses uniformly at random a position $i \in [m(n)]$ of the classical proof to read.
    Then, based
    on his input, the random bits and the value of $y_i$, he chooses $q$ qubits
    from the quantum proof, performs a general POVM measurement on them, and decides on
    acceptance or rejection of the input.
  \end{enumerate}
A promise problem $A=(\ayes,\ano)$ belongs to \PQPCP{q}{\alpha}{\beta} if it has a
  Pointer QPCP proof system with the following properties.
	\begin{description}
		\item[Completeness.] If $x \in \ayes$
			then there exists $y_1...y_{m(n)} \otimes \ket{\psi}$ such that verifier accepts
			with probability at least $\alpha$.
		\item[Soundness.] If $x \in \ano$
			then for all $y_1...y_{m(n)} \otimes \ket{\psi}$ the verifier accepts with probability at most $\beta$.
	\end{description}

\end{definition}

\begin{lemma}
	\label{rem:QMA and PQPCP}
	$\QMA = \PQPCP{q}{\alpha}{\beta}$ where $q = \bigo{1}$
	and $\alpha - \beta \geq 1/\poly{n}$.
\end{lemma}
\begin{proof}
	Since Pointer QPCPs are generalizations of QPCPs, we have that
	$\QPCP{q}{\alpha}{\beta} \subseteq \PQPCP{q}{\alpha}{\beta}$
	for any values of $q$, $\alpha$, and $\beta$.
	From \cref{rem:QMA and QPCP}, it follows that
	$\QMA \subseteq \PQPCP{q}{\alpha}{\beta}$.
	The other direction of the containment follows trivially
	since the \QMA{} verifier can read the whole proof.
\end{proof}

\Cref{conj:weak QPCP verification} asks whether the class \QMA{} also
has Pointer QPCPs with $q = \bigo{1}$ and $\alpha - \beta = \Omega(1)$.


\subsection{The Set Local Hamiltonian Problem}

We define a new \QMA{}-complete problem which is a generalization
of the Local Hamiltonian problem and which will lead to another version
of our conjecture.

\begin{definition}[Set Local Hamiltonian Problem]
  \label{def:SLH}
  The \emph{Set Local Hamiltonian} problem is denoted by \SLH{k}{a}{b}
  where $k \in \Zplus$ is called the locality and for $a,b \in \real$
  it holds that $a<b$.
  It is the following promise problem.
  Let $n$ be the number of the qubits of a quantum system, and $m$ and $l$ be two
  polynomials. The input for the problem are $m(n)$ sets of Hamiltonians.
  For all $i \in \Br{\fn{m}{n}}$ the set $\Hset{H}_i$ contains
  $l(n)$ Hamiltonians, i.e.,
  \[ \forall i \in \Br{\fn{m}{n}} :
  \Hset{H}_i = \cbr{H_{i,1}, \ldots, H_{i,{l(n)}}} . \]
  Each Hamiltonian is positive and has norm at most one, i.e.,
  $\forall i \in \Br{\fn{m}{n}}, \forall j \in \Br{\fn{l}{n}} :
  0 \leq H_{i,j} \leq \id $.
  Each Hamiltonian acts non-trivially on at most $k$ qubits
  out of the $n$ qubits of the quantum system.
  The problem is to decide which one of the following two conditions hold.
  \begin{itemize}
	\item In a YES instance, there exists a function
		$ f : \Br{\fn{m}{n}} \to \Br{\fn{l}{n}} $ and a state
		$\ket{\varphi} \in \complex^{2^n}$ such that
		\[ \bra{\varphi} \sum_{i = 1}^{m(n)} H_{i,f(i)} \ket{\varphi}
		\leq a \cdot m(n) . \]
	\item In a NO instance, for all functions
		$ f : \Br{\fn{m}{n}} \to \Br{\fn{l}{n}} $ and
		for all states $\ket{\varphi} \in \complex^{2^n}$, we have that
		\[ \bra{\varphi} \sum_{i = 1}^{m(n)} H_{i,f(i)} \ket{\varphi}
		\geq b \cdot m(n) . \]
  \end{itemize}
\end{definition}

\begin{lemma}
	\label{rem:SLH and QMA}
	The \SLH{k}{a}{b} problem is \QMA{}-complete for $k \geq 2$
	and $b - a \geq 1/\poly{n}$.
\end{lemma}
\begin{proof}
  For the containment $\SLH{k}{a}{b} \in \QMA$, let the witness have a classical
  part that contains the description of the function $f$ and a quantum part that is
  supposed to be the state \ket{\varphi}.  The quantum verifier can then apply
  the usual eigenvalue estimation on $\sum_{i = 1}^{m(n)} H_{i,f(i)}$.  The
  hardness of \SLH{k}{a}{b} comes trivially from the fact that Local Hamiltonian
  problem is a special case of the Set Local Hamiltonian problem with $\fn{l}{n} = 1$.
\end{proof}

Note that \cref{conj:weak QPCP constraint} asks whether the Set Local Hamiltonian problem remains \QMA-complete
when the locality is constant and the gap between $b$ and $a$ is also constant.


\subsection{CRESP Games}
  \label{sec:cresp}

  We now formally describe a new variant of quantum multi-prover games.
  These games are rather restricted but will allow us to state a third variant of our pointer QPCP conjecture.

  \subsubsection{Description of the Game}
  Let $n \in \Zplus$ be a parameter and $m$ be a polynomial.
  The size of the game will be polynomial in $n$.
  The game is played by one classical prover,
  $\ceil{\log(n+1)}$ quantum provers, and a verifier.
  It is played as follows.
  \begin{enumerate}
    \item \label{step:share the encoding} The quantum provers share the encoding of an arbitrary $n$-qubit state.
      (The encoding maps each qubit into a number of qudits and will be defined later.)
      They are not allowed to share any other resources.
    \item The verifier picks a question $i$ uniformly at random
      out of the $m(n)$ possible questions
      and sends the same question to all the provers (both quantum and classical).
    \item The classical prover replies with $\bigo{\log{n}}$ bits.
    \item Each quantum prover replies by at most $k$ qudits
     from their shared encoded state.
      All the quantum provers use the same strategy.
    \item The verifier accepts or rejects, based on his question and the answers from the provers.
  \end{enumerate}

  We denote these games by the acronym CRESP
  after the \textbf{C}lassical prover, the
  \textbf{R}estricted \textbf{E}ntanglement that the
  quantum provers can share and, since the
  only possible strategy the quantum provers can perform is to swap
  some of their qudits into the message register, we call them
  \textbf{S}wapping-\textbf{P}rovers.

  \begin{figure}[t b]
    \begin{center}
      \begin{tikzpicture}
    \matrix[table,ampersand replacement=\&] (mat11)
    {
      |[fill=red]| \&             \& |[fill=red]|\&              \& |[fill=red]|\&             \& |[fill=red]|\\
                   \& |[fill=red]|\& |[fill=red]|\&              \&             \& |[fill=red]|\& |[fill=red]|\\
                   \&             \&             \& |[fill=red]| \& |[fill=red]|\& |[fill=red]| \& |[fill=red]|            \\
    };

  \SlText{11-1-1}{qubit $1$}
  \SlText{11-1-2}{qubit $2$}
  \SlText{11-1-3}{qubit $3$}
  \SlText{11-1-4}{qubit $4$}
  \SlText{11-1-5}{qubit $5$}
  \SlText{11-1-6}{qubit $6$}
  \SlText{11-1-7}{qubit $7$}

  \CellText{11-1-1}{prover $1$};
  \CellText{11-2-1}{prover $2$};
  \CellText{11-3-1}{prover $3$};
      \end{tikzpicture}
    \end{center}
    \caption{A possible distribution of the encoding of $7$ qubits among $3$ provers. The red cells correspond to the GHZ-like entangled states, while the white cells to $\ket{0}$ states.}
    \label{fig:encoding}
  \end{figure}

  \subsubsection{Restriction on the Entanglement}
  The entangled state the provers share is of the following
  predefined form. First, the provers pick an arbitrary $n$-qubit state $\ket{\phi} \in \complex^{2^n}$.
  The state \ket{\phi} is encoded with a linear isometry
  $\Epsilon = \Epsilon_1 \otimes \Epsilon_2 \otimes \ldots \otimes \Epsilon_n$
  where each qubit of \ket{\phi} is encoded with $\Epsilon_i : \complex^2
  \to \bigotimes_{j=1}^{\ceil{\log(n+1)}} \calH_{i,j}$.
  For all $i$ and $j$, $\calH_{i,j} \cong \complex^4$, that is,
  $\calH_{i,j}$ is a four-dimensional space which we simply call qudit.
  To define $\Epsilon_i$, let's fix some ordering on the non-empty subsets of
  \Br{\ceil{\log(n+1)}}.
  Let $Q_i$ be the $i$-th subset, $\calS_i \eqdef \bigotimes_{j
  \in Q_i} \calH_{i,j}$, and $\overline{\calS_i} \eqdef
  \bigotimes_{j \notin Q_i} \calH_{i,j}$.
  For each $i \in \Br{n}$, we define $\Epsilon_i$ by giving its action
  on the standard basis states.
  \begin{align}
    \fn{\Epsilon_i}{\ket{0}} &\eqdef \frac{1}{\sqrt{2}} \left(\ket{0}^{\otimes \abs{Q_i}} +
  \ket{1}^{\otimes \abs{Q_i}}\right)_{\calS_i}
  \otimes \left(\ket{0}^{
  \otimes \ceil{\log(n+1)} - \abs{Q_i} }\right)_{\overline{\calS_i}}
  \label{eq:encoding on zero} \\
  \fn{\Epsilon_i}{\ket{1}} &\eqdef \frac{1}{\sqrt{2}} \left(\ket{2}^{\otimes \abs{Q_i}} +
      \ket{3}^{\otimes \abs{Q_i}}\right)_{\calS_i}
  \otimes \left(\ket{0}^{
  \otimes \ceil{\log(n+1)} - \abs{Q_i}}\right)_{\overline{\calS_i}}
  \label{eq:encoding on one}
  \end{align}

  We refer to the states in $\calS_i$ as GHZ-like states.
  After $\Epsilon$ is applied, prover $j$ receives the qudits that
  live in space $\bigotimes_{i=1}^{n} \calH_{i,j}$.
  A possible distribution of the qudits is depicted in \cref{fig:encoding}.

\subsubsection{Description of the CRESP Problem}
We are interested in the maximum acceptance probability the provers can
achieve, which is called the value of the game.
Here the maximum is taken over all legitimate shared states and
all legitimate provers' strategies.
We now define the promise problem that corresponds to the approximation
of the value of CRESP games.

\begin{definition}
	Let $k \in \Zplus$ and $\alpha, \beta \in \real$ with $\alpha > \beta$.
	Then, \CRESP{k}{\alpha}{\beta} is the following promise problem.
	The input is the description of a CRESP game defined above where the quantum provers
	answer at most $k$ qudits and the following conditions hold.
	\begin{itemize}
		\item In a YES instance the value of the game is at least $\alpha$.
		\item In a NO instance the value of the game is at most $\beta$.
	\end{itemize}
\end{definition}

We will prove that
the \CRESP{k}{\alpha}{\beta} problem is \QMA{}-complete for $k = \bigo{1}$
	and $\alpha - \beta \geq 1/\poly{n}$.
We defer this proof to \cref{sec:lh to game} as it
needs results that we will establish later.
Again, we note that \cref{conj:weak QPCP game} asks whether \CRESP{k}{\alpha}{\beta} remains \QMA-complete when
$k=O(1)$ and $\alpha-\beta=\Omega(1)$.


\section{Equivalence of Our QPCP Conjectures}
\label{sec:equivalence}
In this section we prove \cref{thm:main}, the equivalence of the three versions of our
Pointer QPCP conjecture.
The proof proceeds in the following three steps.
In \cref{sec:pcp to lh}, we show that if \cref{conj:weak QPCP verification}
is true then \cref{conj:weak QPCP constraint} is also true.
We do this by reducing any problem with a Pointer QPCP proof system
to the Set Local Hamiltonian problem.
In \cref{sec:lh to game} we show that if \cref{conj:weak QPCP constraint}
is true then \cref{conj:weak QPCP game} is also true by giving
a reduction from the Set Local Hamiltonian problem to our
decision problem involving CRESP games.
To complete the cycle, we prove in \cref{sec:game to pcp} that if
\cref{conj:weak QPCP game} is true then \cref{conj:weak QPCP verification}
is also true by giving a Pointer QPCP proof system for an arbitrary CRESP game.
See \cref{fig:arrows of implications} for the arrows of implications.


\subsection{From Pointer QPCP to the Set Local Hamiltonian Problem}
\label{sec:pcp to lh}

In this section we show that if \cref{conj:weak QPCP verification}
is true then \cref{conj:weak QPCP constraint} is also true.
We show that any problem $P \in \PQPCP{q}{\alpha}{\beta}$
is polynomial-time reducible to \SLH{q}{1-\alpha}{1-\beta}.
Assuming \cref{conj:weak QPCP verification}, this means that
the Set Local Hamiltonian problem is \QMA-hard.
The containment of the Set Local Hamiltonian problem
in \QMA{} is implied by \cref{rem:SLH and QMA}.

\begin{theorem}
  Any problem $P \in \PQPCP{q}{\alpha}{\beta}$ can be
  reduced to \SLH{q}{1-\alpha}{1-\beta} in polynomial time.
\end{theorem}
\begin{proof}
  Let $y_1,\ldots,y_m$ be the classical part
  of the proof where $m=m(n)$ for a polynomial $m$ and $\ket{\psi}$ be the
  quantum part of the proof, which contains $p(n)$ qubits, for a polynomial $p$.
  Suppose that each $y_i$ can take $l=l(n)$ different values, for a polynomial $l$.
  We construct an instance of the Set Local Hamiltonian problem
  that consists of $m$ sets of Hamiltonians $\Hset{H}_i$, for $i \in \Br{m}$, where
  $\Hset{H}_i = \{H_{i,j}\}_{j \in [l]} $ and the Hamiltonians act on a
  $p(n)$-qubit system.
  Let $H_{i,j}$ be the rejection POVM element of the Pointer QPCP verifier
  over the constant number of qubits when he reads register $i$ from
  the classical part of the proof and it contains the value $j$, i.e., $j=y_i$.

  First we prove that if there is a proof that makes the Pointer QPCP verifier
  accept with probability greater than $\alpha$ then there is a function $f$
  such that the groundstate of $\sum_{i=1}^m H_{i,f(i)}$ has energy at most
  $(1-\alpha)m$.
  Let $y_1...y_m \otimes \ket{\psi}$ be such proof   and let $\alpha_i$ be the
  acceptance probability of the Pointer QPCP verifier when the verifier queries $i$. Since the verifier picks an  $i$ uniformly at random, it follows that $\frac{1}{m} \sum_i \alpha_i = \alpha$.
  Let $f(i) \eqdef y_i$.
  In this case, the energy of
  $\ket{\psi}$ on $\sum_{i} H_{i,f(i)}$ is
  \[\bra{\psi} \br{\sum_{i} H_{i,y_i}} \ket{\psi}
  \leq \sum_i (1 - \alpha_i) =  (1 - \alpha)m .\]

  For the other direction of the proof, suppose that there is a function $f$ and
  a state \ket{\psi} such that $\bra{\psi} \br{\sum_i H_{i,f(i)}} \ket{\psi}
  \leq \br{1 - \beta}m$.
  Then there is a proof that makes the Pointer QPCP verifier accept with probability
  bigger than $\beta$.
  Let $\br{f(1),f(2),...,f(m)} \otimes \ket{\psi}$ be the proof for the Pointer QPCP
  verifier.
  The acceptance probability of the Pointer QPCP
  verifier with this proof is
  \[	\frac{1}{m}  \sum_i (1 - \bra{\psi}H_{i,f(i)}\ket{\psi})
	= 1 - \frac{1}{m} \bra{\psi} \br{\sum_i H_{i,f(i)}} \ket{\psi}
	\geq \beta. \]
  This finishes the proof of the reduction.
\end{proof}


\subsection{From the Set Local Hamiltonian Problem to CRESP Games}
\label{sec:lh to game}
In this section we show that if \cref{conj:weak QPCP constraint}
is true then \cref{conj:weak QPCP game} is also true.
We do this by giving a reduction from the \SLH{k}{a}{b} problem
to the \CRESP{k}{1 - a/2}{1 - b/2} problem.
Assuming \cref{conj:weak QPCP constraint}, this implies that the
\CRESP{k}{1 - a/2}{1 - b/2} problem is \QMA-hard.
We prove the containment $\CRESP{k}{1 - a/2}{1 - b/2} \in \QMA$ in \cref{lemma:CRESPcomplete}.

We construct a CRESP game for the Set Local Hamiltonian problem.
The main idea in the construction is the following.
In our game, the verifier picks an index $i \in [m]$ uniformly at random
and sends $i$ to all the provers.
The classical prover tells the verifier the specific Hamiltonian
that should be taken from set $i$, i.e., the value of $f(i)$.
The quantum provers share the encoded groundstate of the Hamiltonian
$\sum_i H_{i,f(i)}$ and reply with the
encoding of the qubits that are involved in Hamiltonian $H_{i,f(i)}$.

First, the verifier checks if the
received qudits lie in the codespace of qubit $i$, and if not he rejects.
Using the definition of the encoding, the projector onto the
codespace is described by
\[
  (\Pi_i)_{\calS_i} \otimes \ketbra{0}_{\overline{\calS_i}}
\]
where
\[ \Pi_i =
\frac{1}{2} \br{\sum_{u,v \in \{0,1\}}
  \ket{u^{|Q_i|}}\bra{v^{|Q_i|}} + \sum_{w,z \in
\{2,3\}} \ket{w^{|Q_i|}}\bra{z^{|Q_i|}}}.
\]
Actually, we will see in
\cref{lemma:mixed answer} that it suffices for the verifier to perform only projection $(\Pi_i)_{\calS_i}$.
If the above test succeeds then the verifier picks a bit uniformly at random
and if it is $0$, he accepts.
Otherwise, the verifier decodes the answered qudits by inverting the mapping
$\Epsilon$, defined by
\cref{eq:encoding on zero,eq:encoding on one}, for all the qubits in Hamiltonian
$H_{i,f(i)}$. Then, he performs the measurement that corresponds to $H_{i,f(i)}$
on the decoded qubits and accepts or rejects based on the outcome.

If the Hamiltonian $\sum_i H_{i,f(i)}$ has an eigenstate with small eigenvalue then the provers will pass
the test with high probability. Using the fact that the provers share a state in the predefined encoding and
the restriction on the quantum provers' strategies, we also show that
the verifier will reject with high probability if all states
have high eigenvalues.
The description of the game is in \cref{algo:cresp}.

\begin{algorithm}[t b]
  \caption{CRESP Game for \SLH{k}{a}{b}}
  \begin{enumerate}[itemsep=0.1ex]
    \item The provers pick an $n$-qubit state \ket{\phi} and share its encoding
      \fn{\Epsilon}{\ket{\phi}}.  (In the honest case, \ket{\phi} is supposed to
      be the groundstate of Hamiltonian $H$.)
    \item The verifier picks $i \in \Br{m}$ uniformly at random and sends it to
      all the provers.
    \item The classical prover sends some $j \in \Br{l}$.
    \item Each quantum provers send $k$ qudits.
    \item The verifier performs the following tests.
    \begin{itemize}[itemsep=0.1ex,topsep=0ex]
    \item[$T_1$.] Measure the received qudits that correspond to the space
      $\calS_i$ with the projectors $\{\Pi_i, \Pi_i^\perp\}$ and reject
      if the outcome is $\Pi_i^\perp$. Otherwise, continue.
      \item[$T_2$.]  Pick $b \in \cbr{0,1}$ uniformly at random and accept if $b =
        0$. Otherwise, continue.
      \item[$T_3$.] Decode the received qudits and perform the measurement specified
        by $H_{i,f(i)}$  and accept or reject depending on the
      outcome.
    \end{itemize}
  \end{enumerate}
  \label{algo:cresp}
\end{algorithm}

\begin{theorem}
  \label{thm:hardness cresp}
The game defined by \cref{algo:cresp} has completeness $1 - a/2$ and soundness
$1 - b/2$.
\end{theorem}
\begin{proof}
	\Cref{lem:cresp completeness} proves completeness
	while \cref{lem:cresp soundness} proves soundness.
\end{proof}

\begin{lemma}[Completeness]
	\label{lem:cresp completeness}
  If there is a function $f$ such that the groundstate of $\sum_i
  H_{i,f(i)}$ has eigenvalue at most $a m$ then the maximum
  acceptance probability of the game is at least $1 - a/2$.
\end{lemma}
\begin{proof}
  Let the quantum provers share $\Epsilon(\ket{\psi})$, the encoding of the
  groundstate $\ket{\psi}$ of $H \eqdef \sum_i  H_{i,f(i)}$.
  When the verifier queries $i$, the classical prover answers $f(i)$ and
  all quantum provers honestly reply with their shares of the encodings
  of the $k$ qubits corresponding to $H_{i,f(i)}$.
  The verifier always measures $\Pi_i$ and so he accepts with probability
  \begin{align*}
    \frac{1}{2}+ \frac{1}{2} \left( 1- \frac{1}{m} \sum_{i=1}^m \bra{\psi}H_{i,f(i)}\ket{\psi} \right)
    &=1 -  \frac{1}{2m} \bra{\psi}H\ket{\psi} \geq 1 -  \frac{a}{2}. \myqedhere
  \end{align*}
\end{proof}

The following technical lemma is the key to prove soundness.
It establishes that when the provers reply with the qudits that belong to the
encoding of a different qubit, the verifier will detect it with
probability at least half. We defer the full proof of this lemma to \cref{ap:mixed answer}.

\begin{lemma}
	\label{lemma:mixed answer}
  If the provers are asked for the encoding of qubit $i$ and they answer with
  the qudits that correspond to the encoding of a different qubit, then the
  answered state projects to the correct codespace, i.e., the subspace that
  corresponds to the projector $\Pi_i$, with probability at most $1/2$.
\end{lemma}
\begin{proofsketch}
  Since the provers have the same strategy for a fixed question and they can only
  do swaps, the only cheating strategy for the provers is to answer the encoding of
  a qubit which is different from the one the verifier asked for. In this case, by the properties of the
  chosen encoding, the state that should be a GHZ-like state
  is actually separable and it projects to the correct codespace with probability at most half.
\end{proofsketch}

We prove now the soundness property of the game.

\begin{lemma}[Soundness]
	\label{lem:cresp soundness}
  If, for all functions $f$, the groundstate of
  $\sum_i H_{i,f(i)}$ has eigenvalue at least $b m$ then the maximum
  acceptance probability of the game is at most $1 - b/2$.
\end{lemma}
\begin{proof}
	Let's fix an arbitrary strategy for the classical prover.
	We define the function $f$ by letting $f(i)$ be the answer from
	the classical prover to question $i$.
	Let us also fix the shared state of the quantum provers arbitrarily.
	The strategy of the quantum provers is still left undefined.

	Let $A$ be
	the event that the verifier accepts, $E_i$ be the event that the verifier
	picks question $i$, $T_1$ be the event that the verifier continues after the test
  $T_1$, i.e., the answered qudits were projected onto the codespace, $T_2$ be the
  event that the verifier accepts during test $T_2$, i.e., he picks $b=0$,
  and
	$T_3$ be the event that the verifier accepts on test $T_3$, i.e., in the estimation of
  the energy of the Hamiltonian.
	Let us fix a strategy where the quantum provers answer honestly, i.e.,
	where $\cprob{T_1}{E_i} = 1$.
  The probability of acceptance with this strategy is
	\begin{align*}
		\prob{A} &= \sum_{i=1}^m \prob{E_i} \cdot \cprob{T_1}{E_i}
			\br{\prob{T_2} + \prob{\bar{T_2}} \cdot \cprob{T_3}{E_i,T_1,\bar{T_2}}} \\
		&= \sum_{i=1}^m \frac{1}{m} \br{\frac{1}{2} + \frac{1}{2}
			\cdot \cprob{T_3}{E_i,T_1,\bar{T_2}}} \\
		&\leq \frac{1}{2} + \frac{1}{2} \br{1 - b} \\
		&= 1 - \frac{b}{2}
	\end{align*}
	where the inequality follows from the fact that all the eigenstates of
	$\sum_i H_{i,f(i)}$ have eigenvalue at least $b m$.

  Let $G$ and $B$ be two strategies which are the same except
  when question $r$ is asked for a fixed $r$.
  When asked $r$, the provers following $G$ answer honestly while provers
  following $B$ answer with the encoding of a different qubit.

  We extend the previous notation by adding the superscript of the corresponding
  strategy for the events, e.g., $T_1^G$ is the event that the answered qudits are in
	the correct codespace when the provers follow strategy $G$.
	The following calculation shows that $G$ has higher success probability.
	\begin{align*}
		\prob{A^G} - \prob{A^B} &= \prob{E_k^G} \cdot \cprob{T_1^G}{E_k^G}
			\br{\prob{T_2^G} + \prob{\bar{T_2}^G} \cdot \cprob{T_3^G}{E_k^G,T_1^G,\bar{T_2}^G}} \\
		& \qquad {} - \prob{E_k^B} \cdot \cprob{T_1^B}{E_k^B}
			\br{\prob{T_2^B} + \prob{\bar{T_2}^B} \cdot \cprob{T_3^B}{E_k^B,T_1^B,\bar{T_2}^B}} \\
		&\geq \frac{1}{m} \br{1 \cdot \br{\frac{1}{2} + \frac{1}{2}
			\cdot \cprob{T_3^G}{E_k^G,T_1^G,\bar{T_2}^G}} \right. \\
		& \qquad {} \left. - \half \cdot \br{\frac{1}{2} + \frac{1}{2}
			\cdot \cprob{T_3^B}{E_k^B,T_1^B,\bar{T_2}^B}}} \\
		&\geq \frac{1}{m} \br{\frac{1}{2} - \half} \\
		&= 0
	\end{align*}
	where in the first inequality we used that $\cprob{T_2^B}{E_k^B} \leq 1/2$,
	by \cref{lemma:mixed answer}.
	By a hybrid argument, it is easy to see that the strategy where provers are honest is optimal and so the soundness follows.
\end{proof}

	We now show that even though our game seems very restricted,
	it is in fact \QMA{}-hard to approximate its value to within
	an inverse-polynomial precision.

\begin{theorem}\label{lemma:CRESPcomplete}
The \CRESP{k}{\alpha}{\beta} problem is \QMA{}-complete for $k = \bigo{1}$
	and $\alpha - \beta \geq 1/\poly{n}$.
\end{theorem}

\begin{proof}
  The containment in \QMA{} is simple:
  The \QMA{} proof is the state the provers choose before the encoding
  together with the classical information that
  describes the behavior of all the provers.
  Then the \QMA{} verifier can create the encoding and simulate
  the game.
  This leads to the same acceptance probability as that of the game
  which means that there is an inverse-polynomial gap between completeness
  and soundness in the \QMA{} protocol.

  The \QMA-hardness follows from \cref{rem:SLH and QMA,thm:hardness cresp}.
\end{proof}


\subsection{From CRESP Games to Pointer QPCPs}
\label{sec:game to pcp}
In this section we show that if \cref{conj:weak QPCP game}
is true then \cref{conj:weak QPCP verification} is also true.
We do this by proving that $\CRESP{k}{\alpha}{\beta} \in
\PQPCP{k}{\alpha}{\beta}$.
Assuming \cref{conj:weak QPCP game}, this implies that
$\QMA \subseteq \PQPCP{k}{\alpha}{\beta}$.
The inclusion $\PQPCP{k}{\alpha}{\beta} \subseteq \QMA$
follows trivially, the same way as in \cref{rem:QMA and PQPCP}.

\begin{theorem}
	$\CRESP{k}{\alpha}{\beta} \in \PQPCP{k}{\alpha}{\beta}$.
\end{theorem}
\begin{proof}
  In CRESP games, the strategy of the
  quantum provers consists of the choice of the shared state and the choice of
  which qudits to answer for each one of the verifier's questions.
  For the classical prover, the
  strategy consists of the classical answers for each one of the verifier's questions. Therefore, we can
  have a Pointer QPCP whose proof will be as follows: for the classical part, for each possible question of the verifier, we include the indices of the qudits answered by the quantum provers and the answer of the classical prover. The quantum part of the proof
  will be the shared state before the encoding.
  With this information, the verifier of
  the Pointer QPCP can simulate the classical prover, the quantum provers, and
  the verifier of the CRESP game.

  Formally, the verifier of the Pointer QPCP protocol is provided a proof of the
form $y_1...y_m \otimes \ket{\psi}$,
  where $y_i$ can be seen as a pair $(s_i, c_i)$.
  The verifier will do the following.
  \begin{enumerate}
    \item He picks a question $i$ uniformly at random  as the verifier of the game.
    \item He reads the corresponding strategy of the provers, i.e., $(s_i,c_i)$.
    \item He creates the encoding of the qubits that are specified by
	strategy $s_i$.
    \item He simulates the verifier of the game using the encoded qubits as the quantum
      provers' answers and $c_i$ as the classical prover's answer.
    \item He accepts if and only if the verifier of the game accepts.
  \end{enumerate}

  In our construction, we crucially use the fact that each quantum prover
  has the same strategy, as otherwise, the QPCP verifier would need to read out the strategies
  of each prover, which would require \bigomega{\fn{\log^2}{n}}
bits of information.
  Note that we only read out $k$ qubits from the quantum part of the proof.
  We are left to prove completeness and soundness.

  For completeness, it is not hard to see that if there is a strategy for the
  provers in the game with acceptance probability $p$ then there is a
  Pointer QPCP that accepts with probability $p$ as well, just by
  providing the values of $s_i$, $c_i$, and $\ket{\psi}$ that lead to acceptance
  with probability $p$ in the game.

    For soundness, if there are values of $y_i = (s_i,
    c_i)$ and $\ket{\psi}$ that make the Pointer QPCP verifier accept with
    some probability then these values can be translated to strategies of the provers
    in the CRESP game that will achieve the same acceptance probability.
\end{proof}


\section{Discussions and Open Problems}
\label{sec:open problems}

We defined a new variant of quantum proof systems, the Pointer QPCPs, and provided three equivalent versions of the Pointer QPCP conjecture. Our conjecture is weaker than the original QPCP conjecture and hence may be easier to prove. Moreover, the fact of having an equivalent game formulation might lead to new techniques for resolving the conjecture.

It is an interesting question to see whether we can define a more natural game which is equivalent to the Pointer QPCP conjecture. For our equivalence, we were forced to impose stringent constraints on the game. Nevertheless, it seems that if we allow the quantum provers to
either share some more general entangled state or apply any operator to the state
they share other than swapping, then it is not clear how not to lose the constant gap when constructing the witness \cite{FitzsimonsV15, Ji15} or not to increase the question
size to exponential \cite{NatarajanV15}.
To illustrate this problem, imagine that the provers are
allowed to slightly change the states they return depending on the
questions that were asked.
If the amount of change, in the trace distance, is in the order of \fn{o}{1} then
the verifier will not be able to detect this with constant
probability.\footnote{Let the amount of change be in the order of
\fn{\Theta}{1/\sqrt{n}}, for example.}
When we try to prove soundness by constructing a state with low energy
(as in \cite{FitzsimonsV15}) these errors add up \bigomega{n} times.
The final error can then be too big so we don't know whether the
constructed state is a yes or no instance.

Even with our constraints, CRESP games remain equivalent to Pointer QPCPs.
Since Pointer QPCPs are a superclass of QPCPs, finding a game that is equivalent
to the original QPCP would potentially impose even further constraints.
One of the main problems going from a game back to Local Hamiltonians or QPCPs, is that to simulate the game, the strategies of the provers must be somehow simulated and when we try to do this with Local Hamiltonians, the gap vanishes, while for QPCPs, we require the classical pointer queries.
Note that, in the Set Local Hamiltonian problem the gap doesn't
depend on the size of the sets, by definition.
Whereas, if we want to go to the usual Local Hamiltonian problem then the
absolute gap is divided by the total number of Hamiltonians and so the
gap vanishes.


\section*{Acknowledgements}

We thank Thomas Vidick for helpful discussions. Supported by ERC QCC, ANR RDAM.

\bibliographystyle{halpha}
\bibliography{references}

\appendix

\section{Proof of Lemma~\ref{lemma:mixed answer}}
\label{ap:mixed answer}

\begin{proofof}{lemma:mixed answer}
  We remind here the notation  from \cref{sec:cresp}.
  The encoding of qubit $k$ can be split into two parts: a GHZ-like
  state and copies of $\ket{0}$.
  (For the formal definition see \cref{eq:encoding on zero,eq:encoding on one}.)
  For the encoding of qubit $k$, let $Q_k$ be the subset of provers that receive a share of the GHZ-like state
  and $\calS_k \cong \complex^{4^{|Q_k|}}$ its corresponding subspace.
  The projection over $\calS_k$ onto the codespace of the encoding of qubit
  $k$ is
  \[ \Pi_k = \frac{1}{2} \br{\sum_{u,v
  \in \{0,1\}} \ket{u^{|Q_k|}}\bra{v^{|Q_k|}} + \sum_{w,z \in
  \{2,3\}} \ket{w^{|Q_k|}}\bra{z^{|Q_k|}}}.  \]

  \begin{figure}[t b]
   \centering
   \begin{subfigure}[t]{0.3\textwidth}
    \begin{center}
      \begin{tikzpicture}
    \matrix[table,ampersand replacement=\&] (matq11)
    {
      |[fill=red]|\\
      |[fill=red]|\\
      |[fill=red]|\\
                  \\
                  \\
    };

    \matrix[table,right=1ex of matq11, ampersand replacement=\&] (matq12)
    {
      |[fill=red]|\\
      |[fill=red]|\\
      |[fill=red]|\\
      |[fill=red]|\\
                  \\
    };

    \SlTextNormal{q11-1-1}{$i$}
    \SlTextNormal{q12-1-1}{$j$}
    \CellText{q11-1-1}{prover $1$};
    \CellText{q11-2-1}{prover $2$};
    \CellText{q11-3-1}{prover $3$};
    \CellText{q11-4-1}{prover $4$};
    \CellText{q11-5-1}{prover $5$};
    \end{tikzpicture}
    \end{center}
   \caption{$Q_i \subset Q_j$}
   \label{fig:subset}
 \end{subfigure}
 \begin{subfigure}[t]{0.3\textwidth}
    \begin{center}
      \begin{tikzpicture}
    \matrix[table,ampersand replacement=\&] (matq21)
    {
      |[fill=red]|\\
      |[fill=red]|\\
      |[fill=red]|\\
                  \\
                  \\
    };

    \matrix[table,right=1ex of matq11, ampersand replacement=\&] (matq22)
    {
                  \\
      |[fill=red]|\\
      |[fill=red]|\\
                  \\
                  \\
    };

    \SlTextNormal{q21-1-1}{$i$}
    \SlTextNormal{q22-1-1}{$j$}
    \CellText{q21-1-1}{prover $1$};
    \CellText{q21-2-1}{prover $2$};
    \CellText{q21-3-1}{prover $3$};
    \CellText{q21-4-1}{prover $4$};
    \CellText{q21-5-1}{prover $5$};
    \end{tikzpicture}
    \end{center}
    \caption{$Q_j \subset Q_i$}
    \label{fig:not subset}
  \end{subfigure}
 \begin{subfigure}[t]{0.3\textwidth}
    \begin{center}
      \begin{tikzpicture}
    \matrix[table,ampersand replacement=\&] (matq31)
    {
      |[fill=red]|\\
      |[fill=red]|\\
      |[fill=red]|\\
                  \\
                  \\
    };

    \matrix[table,right=1ex of matq31, ampersand replacement=\&] (matq32)
    {
      \\
      |[fill=red]|\\
      |[fill=red]|\\
      |[fill=red]|\\
                  \\
    };

    \SlTextNormal{q31-1-1}{$i$}
    \SlTextNormal{q32-1-1}{$j$}
    \CellText{q31-1-1}{prover $1$};
    \CellText{q31-2-1}{prover $2$};
    \CellText{q31-3-1}{prover $3$};
    \CellText{q31-4-1}{prover $4$};
    \CellText{q31-5-1}{prover $5$};
    \end{tikzpicture}
    \end{center}
    \caption{$Q_i \not\subset Q_j$ and $Q_j \not\subset Q_i$}
    \label{fig:disjoint}
  \end{subfigure}
  \caption{Examples of the cases for the expected vs.\ received encodings.
	The red cells correspond to shares
	of GHZ-like states and the white cells correspond to copies of
	$\ket{0}$.}
  \end{figure}
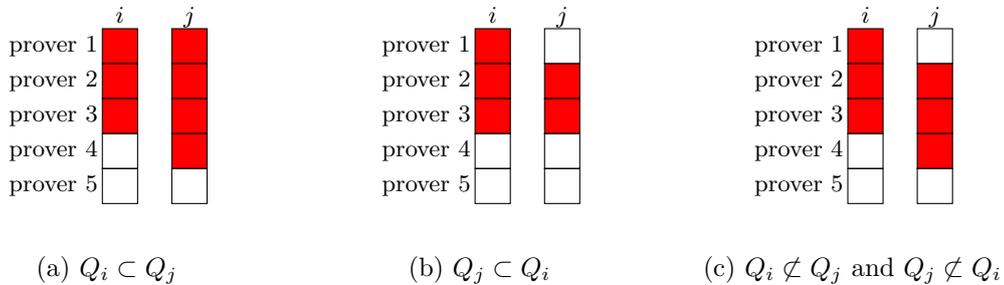

Let $i$ be the qubit whose encoding was asked. Since the provers are dishonest
  and follow the same strategy, they all answer with the encoding of a qubit $j
  \ne i$.

  We split the argument into two cases.
  If $Q_i \subset Q_j$ then the qudits that should contain the GHZ-like state on
  the encoding of $i$ contain the reduced state of a bigger GHZ-like state
  and the density matrix of this reduced state is of the form
  \[\rho_1 =
    \frac{|\alpha|^2}{2} \br{\kb{0^{|Q_i|}} + \kb{1^{|Q_i|}}}
    + \frac{|\beta|^2}{2} \br{\kb{2^{|Q_i|}} + \kb{3^{|Q_i}}}
\]
  for some $\alpha$ and $\beta$, depending on the $j$-th qubit, with
  $|\alpha|^2 + |\beta|^2 = 1$. This case is depicted in \cref{fig:subset}.
  The probability that this state projects onto
  $\Pi_i$ is
  \begin{align*}
       \trace{\Pi_i \rho_1}
       &= \frac{|\alpha|^2}{2} \br{\trace{\Pi_i\kb{0^{|Q_i|}}} +
  \trace{\Pi_i\kb{1^{|Q_i|}}}} \\
  & \qquad {} + \frac{|\beta|^2}{2} \br{\trace{\Pi_i\kb{2^{|Q_i|}}} + \trace{\Pi_i\kb{3^{|Q_i|}}}}\\
       &=
     \frac{|\alpha|^2}{2} \br{\half + \half}
   + \frac{|\beta|^2}{2} \br{\half + \half} \\
       &= \half
   \end{align*}
   where the second equality comes from the fact that  $\norm{\Pi_i\ket{i^{|Q_i|}}}^2 =
   \half$ for $i \in \{0,1,2,3\}$.

  If $Q_i \not\subset Q_j$ then the set $Q_i \setminus Q_j$ is non-empty.
  Let $|Q_i \setminus Q_j| = r$ and $|Q_i \cap Q_j| = s$.
  The provers in  $Q_i \setminus Q_j$ answer $\ket{0}$ and the remaining provers
  send either a GHZ-like state (when $Q_j \subset Q_i$)  or a
  reduced state of a GHZ-like state (when $Q_j \not\subset Q_i$). These
  cases are depicted in \cref{fig:not subset,fig:disjoint}.
  Therefore, the answer from the provers has the form (up to some permutation of
  the qubits, but since the projection is symmetric, we can consider any
  arbitrary order)
  \[\rho_2 =
    \kb{0^r} \otimes \sigma
\]
   for some density matrix $\sigma$ on $s$ qudits.

   From the structure of $\Pi_i$ and using that $|Q_i| = r+s$, we have that
   $\norm{\Pi_i\ket{0^r}\otimes\ket{u^s}} = 0$ for $u \in \{1,2,3\}$. Let
   $\sigma_0 \eqdef \kb{0^s} \sigma \kb{0^s}$ be the projection of $\sigma$ onto
   $\ketbra{0^s}$. It
   follows that
 the probability that the answered state projects
  onto $\Pi_i$ is
     \begin{align*}
       \trace{\Pi_i \rho_2} &=
       \trace{\Pi_i\kb{0^r}\otimes \sigma} \\
       &= \trace{\Pi_i\kb{0^r}\otimes \sigma_0} \\
       &\leq\trace{ \Pi_i \kb{0^{r+s}}} \\
       &= \half.
   \end{align*}
   Therefore, the answered state projects onto the codespace with probability at
   most $1/2$.
 \end{proofof}
\end{document}